\documentclass{llncs}

\usepackage{amssymb}
\usepackage{amsmath}
\usepackage{dsfont}
\usepackage{xspace}
\usepackage{graphicx}
\usepackage{mathabx}
\usepackage{times}
\usepackage{hyperref}
\usepackage{smallsec}
\usepackage{floatrow}
\newfloatcommand{capbtabbox}{table}[][\FBwidth]

\usepackage{algpseudocode}

\newcommand{\zug}[1]{\langle #1  \rangle}
\newcommand{\set}[1]{\{ #1  \}}
\newcommand{\stam}[1]{}

\stam{
\setlength{\evensidemargin}{0in}
\setlength{\oddsidemargin}{0in}
\setlength{\textwidth}{6.2in}
\setlength{\textheight}{8.3in}
\setlength{\topmargin}{0in}
\setlength{\headheight}{0in}
\setlength{\headsep}{.4in}
\setlength{\itemsep}{-\parsep}

\setlength{\parskip}{\smallskipamount}
}

\newcommand{\A}{{\cal A}}

\newcommand{\C}{{\cal C}}
\newcommand{\D}{{\cal D}}

\newcommand{\F}{{\cal F}}

\newcommand{\N}{{\cal N}}
\newcommand{\M}{{\cal M}}
\renewcommand{\O}{{\cal O}}

\renewcommand{\S}{{\cal S}}
\renewcommand{\P}{{\cal P}}

\renewcommand{\tt}{\texttt{tt}}
\newcommand{\ff}{\texttt{ff}}

\newcommand{\Forw}{\mbox{\sc Forward}}

\newcommand{\Nat}{\mathds{N}}

\newcommand{\Real}{\mathds{R}}

\newcommand{\score}{\mbox{\sc Score}}

\stam{
\newtheorem{theorem}{Theorem}[section]

\newtheorem{lemma}[theorem]{Lemma}

\newtheorem{definition}{Definition}[section]

\newtheorem{rmrk}[theorem]{Remark}
\newtheorem{exmpl}[theorem]{Example}
\newenvironment{remark}{\begin{rmrk}\rm}{\hspace{\stretch{1}}\end{rmrk}}
\newenvironment{example}{\begin{exmpl}\rm}{\hspace{\stretch{1}}\end{exmpl}}

\def\eod{\vrule height 6pt width 5pt depth 0pt}
\newenvironment{proof}{\noindent {\bf Proof:} \hspace{.677em}}
                      {\hspace*{\fill}{\eod}}
}

\title{Computing Scores of Forwarding Schemes in \\Switched Networks with Probabilistic Faults\thanks{This research was supported in part by the Austrian Science Fund (FWF) under grants S11402-N23 (RiSE/SHiNE) and Z211-N23 (Wittgenstein Award) and by the People Programme (Marie Curie Actions) of the European Union's Seventh Framework Programme FP7/2007-2013/ under REA grant agreement 607727.
}}
\author{Guy Avni\inst{1} \and Shubham Goel\inst{2} \and Thomas A. Henzinger\inst{1} \and Guillermo Rodriguez-Navas\inst{3}}
\institute{IST Austria\and IIT Bombay\and M\"alardalen University}

\begin{document}
\maketitle
\begin{abstract}
Time-triggered switched networks are a deterministic communication infrastructure used by real-time distributed embedded systems. 
Due to the criticality of the applications running over them, developers need to ensure that end-to-end communication is dependable and predictable. Traditional approaches assume static networks that are not flexible to changes caused by reconfigurations or, more importantly, faults, which are dealt with in the application using redundancy.
We adopt the concept of handling faults in the switches from non-real-time networks while maintaining the required predictability.

We study a class of forwarding schemes that can handle various types of failures. We consider probabilistic failures. For a given network with a forwarding scheme and a constant $\ell$, we compute the {\em score} of the scheme, namely the probability (induced by faults) that at least $\ell$ messages arrive on time. We reduce the scoring problem to a reachability problem on a Markov chain with a ``product-like'' structure. Its special structure allows us to reason about it symbolically, and reduce the scoring problem to \#SAT. Our solution is generic and can be adapted to different networks  and other contexts. Also, we show the computational complexity of the scoring problem is \#P-complete, and we study methods to estimate the score. We evaluate the effectiveness of our techniques with an implementation.
\end{abstract}

\section{Introduction}
An increasing number of distributed embedded applications, such as the Internet-of-Things (IoT) or modern Cyber-Physical Systems, must cover wide geographical areas and thus need to be deployed over large-scale switched communication networks. The switches used in such networks are typically fast hardware devices with limited computational power and with a global notion of discrete time. 
Due to the criticality of such applications, developers need to ensure that end-to-end communication is dependable and predictable, i.e. messages need to arrive at their destination on time. 
The weakness of traditional {\em hard} real-time techniques is that they assume nearly static traffic characteristics and {\em a priori} knowledge about them. These assumptions do not fit well with setups where highly dynamic traffic and evolving network infrastructure are 
the rule and not the exception; e.g. see \cite{Shree:2013}. For this reason, there is a pressing need to combine flexibility and adaptability features with traditional hard real-time methods 
\cite{Ferreira:2006,GarciaLopezVillar:2013,Guti:2015}.

The Time-Triggered (TT) scheduling paradigm has been advocated for real-time communication over switched networks~\cite{steiner2014towards}. The switches follow a static {\em schedule} that prescribes which message is sent through each link at every time slot. The schedule is synthesized offline, and it is repeated cyclically during the system operation~\cite{pozo2015smt}. TT-schedules are both predictable and easy to implement using a simple lookup table. Their disadvantage is that they lack {\em robustness}; even a single fault can cause much damage (in terms of number of lost messages). Error-handling is left to the application designer and is typically solved by  statically introducing redundancy~\cite{bauer2000transparent}. Static allocation of redundancy has its limitations: i) it adds to the difficulty of finding a TT-schedule, which is a computationally demanding problem even before the addition of redundant messages, and ii) it reduces the effective utilization of resources. 

In contrast, non real-time communication networks typically implement error-recovery functionality within the switches, using some kind of flexible routing, to reduce the impact of crashes. Such an approach is used in {\em software defined networking} (SDN)  \cite{KR+15}, which is a booming field in the context of routing in the Internet. Handling crashes has been extensively studied in such networks (c.f., \cite{YLS+14,CGM+16,RCGF13} and references therein), though the goal is different than in real-time networks; a message in their setting should arrive at its destination as long as a path to it exists in the network. Thus, unlike real-time applications, there is no notion of a ``deadline'' for a message.

In this work we explore the frontier between both worlds. We adopt the concept of programmable switches from SDN to the real-time setting in order to cope with network faults. The challenge is to maintain the predictability requirement, which is the focus of this work. We suggest a class of deterministic routing schemes, which we refer to as {\em forwarding schemes}, and we show how to predict the behavior of the network when using a particular forwarding scheme. More formally, the input to our problem consists of a network $\N$ that is accompanied with probabilities of failures on edges, a set of messages $\M$ to be routed through $\N$, a (deterministic) forwarding scheme $\F$ that is used to forward the messages in $\M$, a timeout $t\in \Nat$ on the arrival time of messages, i.e., if a message arrives after time $t$, it is considered to be lost, and a guarantee $\ell \in \Nat$ on the number of messages that should arrive. Our goal is to compute the {\em score} of $\F$, which is defined as the probability (induced by faults) that at least $\ell$ messages arrive at their destinations on time when forwarding using $\F$.

Our score is a means for predicting the outcome of the network. If the score is too low, a designer can use redundancy techniques to increase it. Also, it is a means to compare forwarding schemes. When constructing a forwarding scheme, be it a TT-schedule or any other scheme, a designer has control on some of the components and others are fixed by the application. For example, in many networks, the size of the switches' queues are fixed to be small, making it impossible to use algorithms that rely on large memories.  As a second example, the message priorities are often fixed by their criticality. 
The choices made by the designer can highly influence the performance of the system on the one hand, and are very hard to predict on the other; especially when faults come into the picture. Our score can be used to compare different forwarding schemes, allowing the designer to evaluate his algorithm of choice. Also, our solution can be used for sensitivity analysis with respect to certain parameters of the network;
for example, one can fix the desired score of a scheme, and compute the threshold $\ell$ that guarantees this score, or 
 the score and $\ell$, and find the error probabilities for the channels \cite{AK15}.

A first step towards handling faults in the switches was made in \cite{AGR16}. In their framework, the switches follow a TT-schedule and resort to a forwarding algorithm once a crash occurs. Our forwarding scheme is simpler and allows consideration of richer faults in a clean and elegant manner, which were impossible to handle in \cite{AGR16}'s framework. More importantly, they study adversarial faults whereas we study probabilistic ones, which are a better model for reality while they are considerably more complicated to handle. Using failover paths to allow for flexibility in switched networks was considered in \cite{liu2013f10,wei2014exploiting}.

The definition of the class of forwarding schemes requires care. On the one hand, the switches computation power is limited, so forwarding rules in the switches should be specified as propositional rules. But, on the other hand, it is infeasible to manually specify the rules at each switch as the network is large and is subject to frequent changes. So, we are required to use a central symbolic definition of an algorithm. However, while the definition of the central algorithm uses propositional rules, it should allow for variability between the switches and the messages' behavior in them. There are many ways to overcome these challenges, and we suggest one solution, which is simple and robust. Our forwarding scheme consists of three components. The first component is a {\em forwarding algorithm} that the switches run and is given by means of propositional forwarding rules. The two other components allow variability between the switches, each switch has priorities on messages, and each message has a preference on outgoing edges from each switch. The forwarding rules of the algorithm take these priorities and preferences into consideration. A similar priority-list model is taken in \cite{HPSK16}. Our algorithm for computing the score of a scheme is general and can handle various forwarding schemes that are given as propositional rules as we elaborate in Section~\ref{sec:disc}.

In order to score a given forwarding scheme, we first reduce the scoring problem to a reachability problem on a certain type of Markov chain, which is constructed in two steps. First, we focus on an individual message $m$ and construct a deterministic automaton $\D_m$ that simulates the forwarding scheme from the perspective of the message. Then, we combine the automata of all the messages into an automaton that simulates their execution simultaneously, and construct a Markov chain $\C$ on top of it by assuming a distribution on input letters (faults). The size of $\C$ is huge and the crux of our approach is reasoning about it symbolically rather than implicitly using PRISM \cite{KNP11} for example. We construct a  Boolean formula $\psi$ that simulates the execution of $\C$. The special product-like structure of $\C$ allows us to construct $\psi$ that is proportional in size to the sum of sizes of the $\D_m$ automata rather than the product of their sizes, which is the size of $\C$. There is a one-to-one correspondence between satisfying assignments to $\psi$ and ``good outcomes'', namely outcomes in which at least $\ell$ messages arrive on time. We then infer the score of the forwarding scheme from the {\em weighted} count of satisfying assignments to $\psi$; the weight of a satisfying assignment is the probability of the crashes in the corresponding execution of the network.

The problem of counting the number of satisfying assignments of a Boolean formula is called \#SAT and it has received much attention. The practical developments on this problem are quite remarkable given its computational intractably; even deciding whether a Boolean formula has one solution is an NP-complete problem that was considered impossible to solve practically twenty years ago, {\it a fortiori} counting the number of solutions of a formula, which is a \#P-complete problem and ``closer'' to PSPACE than to NP. Still, there are tools that calculate an exact solution to the problem \cite{Thu06} and a recent line of work that adapts the rich theory of finding approximate solutions with high probability \cite{JVV86} to practice (see \cite{MVC+16} for an overview). Also, extensions of the original problem were studied; strengthening of the formula to SMT rather than SAT \cite{CDM15} and reasoning about assignments with weights, referred to {\em weighted} \#SAT. As mentioned above, our solution requires this second extension. We show that we can alter the formula we construct above to fit in the framework of \cite{CFMV15}, allowing us to use their reduction and generate an equivalent  \#SAT instance.

While solving \#SAT is becoming more practical, it is still far from solved and it would be surprising if the tools will ever be able to compete with tools for solving SAT, e.g.,  \cite{MB08}. 
Thus, one can question our choice of using such a heavy tool to solve our scoring problem. We show that a heavy tool is essential by showing that scoring a forwarding scheme is \#P-complete, by complementing the upper bound above with a reduction in the other direction: from \#SAT to scoring a forwarding scheme.

We also study approaches to estimate the score of a forwarding scheme. We run a randomized algorithm that, with high probability, finds a solution that is close to the actual score. Using an approximate counting tool to count the Boolean formula we construct above, performs very poorly as the reduction of \cite{CFMV15} constructs an instance which is particularly hard for the approximate counting techniques. Thus, in order to employ the tools to approximately solve \#SAT we need to bypass the reduction. We suggest an iterative algorithm that takes advantage of the fact that in practice, the probability of failure is low, so traces with many faults have negligible probability. A second technique we use  is a Monte-Carlo simulation, which has been found very useful in reasoning about networks \cite{rubinstein2011simulation} as well as in statistical model checking in tools like PLASMA \cite{JLS12}, UPPAAL \cite{LPY97}, and PVeSta \cite{AM11}.

We have implemented all our techniques. We show that the exact solution scales to small networks. The solution that relies on approximated counting scales better, but is overshadowed by the Monte-Carlo approach, which scales nicely to moderate networks. We also use the exact solution  to evaluate the scores of the Monte-Carlo approach and we find that it is quite accurate. We note that our counting techniques rely on counting tools as black-boxes and, as mentioned above, improving these techniques is an active line of work. We expect these tools to improve over time, which will in turn improve the scalability of our solution.

\section{Preliminaries}
\label{sec:prelim}
We model a network as a directed graph $\N = \zug{V, E}$.
For a vertex $v \in V$, we use $out(v) \subseteq E$ to denote the set of outgoing edges from $v$, thus $out(v) = \set{\zug{v, u} \in E}$. A collection $\M$ of messages are sent through the network.
Each message $m \in \M$ has a source and a target vertex, which we refer to as $s(m)$ and $t(m)$, respectively. Time is discrete. There is a global timeout $t \in \Nat$ and a message meets the timeout if it arrives at its destination by time $t$.

\subsection*{Forwarding messages}
A {\em forwarding scheme} is a triple $\F = \zug{\A, \set{\prec_v}_{v \in V}, \set{\prec_m^v}_{m \in \M, \ v \in V}}$, where $\A$ is a {\em forwarding algorithm} that the switches run and we describe the two other components below. For ease of notation, we assume the same number of edges $d \in \Nat$ exit all the switches in the network and in each switch they are ordered in some manner\footnote{In many settings, messages are grouped into few priorities making ``priority ties'' common. We assume a total order on message priorities, i.e., there is some arbitrary procedure to break ties.}. Then, our rules forward messages with respect to this order. For example, we can specify a rule that says ``forward a message $m$ on the first edge'' by writing $\Forw(m, e_1)$. The two other components of $\F$ allow variability; each switch $v \in V$ has an order $\prec_v$ on messages, which are priorities on messages, and each message $m \in \M$ has an ordering $\prec_m^v$ on the outgoing edges from $v$, which are preference on edges. 

The propositional rules in $\A$ are of the form $\varphi \rightarrow \Forw(m,e)$. We refer to $\varphi$ as the assertion of the rule and its syntax is as follows
\[\varphi \ ::= \ m \ | \ e_i \ | \  m < m' \ | \ e_i <_m e_j  \ | \ \varphi \vee \varphi \ | \ \neg \varphi
\]
Note that $m$ and $m'$ refer to specific messages in $\M$ while $e_i$ refers to the $i$-th exiting edges from a switch. The forwarding at a switch is determined only by the local information it has; the messages in its queue and its outgoing active edges. In other words, switches are not aware of faults in distant parts of the network and this fits well with the philosophy of the simple networks we model. 

Intuitively, the algorithm takes as input the messages in the queue as well as the active edges, and the output is the forwarding choices. Accordingly, the semantics of an assertion $\varphi$ is with respect to a set of messages $M \subseteq \M$ (the messages in the queue) and a set of edges $T \subseteq \set{e_1,\ldots, e_d}$ (the active edges). Consider a rule $\varphi \rightarrow \Forw(m,e_i)$. We denote by $(M, T) \models_{\prec_v, \set{\prec_m^v}_{m \in \M}} \varphi$ the fact that $(M,T)$ satisfies $\varphi$. Then, $m$ is forwarded on the $i$-th outgoing edge from $v$, namely $e_i$. When $\prec_v$ and $\prec_m^v$ are clear from the context, we omit them. The semantics is defined recursively on the structure of $\varphi$. For the base cases, we have $(M, T) \models m$ iff $m \in M$, thus $m$ is in $v$'s queue, we have $(M,T) \models e_i$ iff $e_i \in T$, thus $e_i$ is active, we have $(M,T) \models (m < m')$ iff $m \prec_v m'$, thus $m'$ has precedence over $m$ in $v$, and we have $(M,T) \models (e_i <_m e_j)$ iff $e_i \prec_m^v e_j$, thus $m$ prefers being forwarded on the $j$-th edge over the $i$-th edge. The inductive cases are as expected.

The algorithm forwards messages on active links. We think of its output as pairs $O \subseteq \M \times E$, where $\zug{m, e} \in O$ implies that the algorithm forwards $m$ on $e$. We require that the algorithm obey the constraints of the network; at most one message is forwarded on a link, messages are forwarded only on active links, messages originate only from their source switch, they are forwarded only after they are received, and they are not forwarded from their destination.

It is sometimes convenient to use definitions of sets in an algorithm as we illustrate in the examples below. A definition of a set is either a collection of messages or a collection of edges that satisfy an assertion as in the above. We also allow set operations like union, intersection, and difference, for sets over the same types of elements. Later on, when we simulate the execution of the forwarding algorithm as a propositional formula, we use extra variables to simulate these operations.

\begin{example} \label{ex:TT} {\bf TT-schedule}
A time-triggered schedule (TT-schedule, for short) assigns messages to edges such that (1) the schedule assigns a message $m$ on a path from its source to target, i.e., it is not possible that $m$ is scheduled on $e$ before it reaches $s(e)$, (2) two messages cannot be sent on the same link at the same time, and (3) all messages must arrive by time $t$. Given a TT-schedule $S$, we can construct an equivalent forwarding scheme assuming there is no {\em redundant waiting}, namely assuming a message $m$ arrives at a switch $v$ at time $i$ and should be forwarded on $e$ at a later time, then, if $m$ stays in $v$, it is only because $e$ is occupied by a different message. We note that a schedule induces an order on the messages  at each vertex, which we use as $\prec_m$, and it induces a path $\pi_m$ for each message, which induces an order $\prec^v_m$ in which the edges on $\pi_m$ have the highest preference.

In order to describe the rules of the algorithm (as well as the rules in the following example), we introduce several definition. For $S \subseteq \M$, we define an assertion $\mbox{priority}(m, S)$ that is satisfied in switches where $m$ has the highest priority out of the messages in $S$, thus $\mbox{priority}(m, S) = \bigwedge_{m' \in S} (m' < m)$. Next, we define an assertion $\mbox{prefers}(m, e_i)$ that is satisfied in vertices where $m$ prefers $e_i$ over all the active edges, thus $\mbox{prefers}(m, e_i) = \bigwedge_{j \neq i} \big(e_j \rightarrow (e_j <_m e_i)\big)$. Finally, we define a set of message $S_{e_i} = \set{m \in \M: \mbox{prefers}(m, e_i)}$, namely $S_{e_i}$ at a vertex $v$ contains the messages that are forwarded on $i$-th outgoing edge from $v$. 

We are ready to describe the algorithm using forwarding rules. For every $m \in \M$ and $i=1,\ldots,d$, $m$ is forwarded on $e_i$ when (1) $m$ is in the queue, (2) $m$ prefers $e_i$, (3)~$e_i$ is active, and (4) $m$ has the highest priority of the messages in $S_{e_i}$. The corresponding rule is $m \wedge \mbox{prefers}(m, e_i) \wedge e_i \wedge \mbox{priority}(m, S_{e_i})\rightarrow \Forw(m, e_i)$.
\end{example}

\begin{example}
\label{ex:hot-potato}
{\bf Hot-potato} This algorithm is intended for networks in which the switches' queue size is limited. Intuitively, messages are ordered in decreasing priority and are allowed to choose free edges according to their preferences. So, assume that the set of active outgoing edges of a switch $v$ is $T \subseteq \set{e_1,\ldots, e_d}$, and the message in the queue are $M = \set{m_1, \ldots, m_k}$ ordered in increasing priority, i.e., for $1 \leq i < j \leq k$, we have $m_j \prec_v m_i$. Then, $m_1$ chooses its highest priority edge $e$ in $T$, i.e., for every other edge $e' \in T$, we have $e' \prec_m^v e$. Following $m_1$, the message $m_2$ chooses its highest priority edge in $T \setminus \set{e}$, and so forth. If a message is left with no free outgoing edge, it stays in $v$'s queue. The algorithm has a low memory consumption: rather than keeping a message $m$ in the queue till its preferred edge is free, the switch forwards $m$ on a lower-preference edge. Note that unlike the algorithm in Example~\ref{ex:TT}, the hot-potato algorithm has fault tolerant capabilities. 

We alter the assertion prefers by adding to it a set $T \subseteq \set{e_1, \ldots, e_d}$, so that $\mbox{prefers}(m, e, T)$ is satisfies when $m$ prefers $e$ over all active edges in the set $T$. Let $k = \min\set{|\M|, d}$, thus $k$ is an upper bound on the number of messages that can be forwarded from a switch at each time.  We define sets of messages $S_1 \supseteq S_2 \supseteq \ldots \supseteq S_k$, where $S_1$ is the set of messages in the queue, and for $1 < i \leq k$, we have $S_i = S_{i-1} \setminus \set{m :\bigvee_{m \in S_i} \mbox{priority}(m, S_i)}$. That is, for $1 \leq i \leq |M|$, the set $S_i$ contains the messages after the $i$-highest priority messages have been forwarded. Now, we define a sequence of $k$ sets of edges $T_1 \supseteq \ldots \supseteq T_{k}$. The set $T_i$ is the set of edges that are available for the message of priority $i$. Thus, we have $T_1$ is the set of active edges, and, for $i > 1$, the set $T_i$ contains the edges in $T_{i-1}$ minus the edge that the message of priority $i-1$ selected, thus $T_i = T_{i-1} \setminus \set{e_j: \bigvee_{m \in \M} \mbox{priority}(m, S_i) \wedge \mbox{prefers}(m, e, T_i)}$. Finally, for every $m \in \M$ and $1 \leq i \leq k$, and $1 \leq j \leq d$, if $m$ is the highest priority message in $S_i$, and $e_j$ is its highest priority edge in $T_i$, we forward $m$ on $e_j$, thus we have a rule $\mbox{priority}(m, S_i) \wedge \mbox{prefers}(e_j, T_i) \rightarrow \Forw(m, e_j)$.
\end{example}

\subsection*{Faults and Outcomes}
We consider two types of faults. The first type are crashes of edges. We distinguish between two types of crashes: {\em temporary} and {\em permanent} crashes in which edges can and cannot recover, respectively. A second type of fault model we consider are faults on sent messages. We consider {\em omissions} in which a sent message can be lost. We assume the switches detect such omissions, so we model these faults as a sent message that does not reach its destination and re-appears in the sending switch's queue. As we elaborate in Section~\ref{sec:disc}, our approach can handle other faults such as ``clock glitches'', which are common in practice. 

The {\em outcome} of a forwarding scheme $\F$ is a sequence of {\em snapshots} of the network at each time point. Each snapshot, which we refer to as a {\em configuration}, includes the positions of all the messages, thus it is a set of $|\M|$ pairs of the form $\zug{m,v}$, meaning that $m$ is on vertex $v$ in the configuration. We use $\O$ to denote the set of all outcomes. Each outcome in $\O$ has $t+1$ configurations, thus $\O \subseteq (\M \times V)^{t+1}$. All outcomes start from the same initial configuration $\set{\zug{m, s(m)}: m \in \M}$ in which all messages are at their origin. Consider a configuration $C$. Defining the next configuration $C'$ in the outcome is done in two steps.  In the first step, we run $\F$ in all vertices. Consider a vertex $v$, let $T \subseteq out(v)$ be a set of active edges. The set of messages in $v$'s queue is $M = \set{m: \zug{m,v} \in C}$. Intuitively, we run $\F$ at $v$ with input $M$ and $T$. The forwarding algorithm keeps some of the messages $S \subseteq M$ in $v$'s queue and forwards others. The messages in $S$ stay in $v$'s queue, thus we have $\zug{m,v} \in C'$ for every message $m\in S$. Recall that the algorithm's output is $O \subseteq (\M \times E)$, where $\zug{m, e} \in O$ means that $m$ is forwarded on the link $e$. In the second step, we allow omissions to occur on the pairs in $O$. If an omission occurs on $\zug{m, \zug{v,u}} \in O$, then $m$ returns to the source of the edge and we have $\zug{m,v} \in C'$, and otherwise, sending is successful and we have $\zug{m, u} \in C'$. 

\stam{
We refer to the sequence of failures in an outcome as a {\em failure sequence}. It is a sequence of the form $T_1, L_1, \ldots, T_t, L_t$, where, for $1 \leq i \leq t$, the set of edges $T_i$ are the active edges at time $i$, and the set $L_i \subseteq (\M \times E)$ represents the messages that are successfully forwarded at time $i$. Thus, the edges in $E \setminus T_i$ are the edges that crashed at time $i$. Since we focus on crashes, we have $T_i \subseteq T_{i+1}$, for $1 \leq i < t$. Assuming the forwarding scheme outputs $O_i$ at time $i$, we have $L_i \subseteq O_i$. Again, the messages in $O_i \setminus L_i$ are the messages that got omitted. We say that $k$ edge-crashes occur if $|E \setminus T_t| = k$ and $k$ omissions occur if $\sum_{i=1}^t |O_i \setminus L_i| = k$. In \cite{AGR16}, we considered the adversarial setting in which we expected the ``worst'' fault sequence to occur once a forwarding scheme is fixed. More formally, we allowed only edge crashes, thus no omissions can occur, and we asked, given a forwarding scheme and two thresholds $k$ and $\ell$, is there a fault sequence with at most $k$ crashes in which less than $\ell$ messages arrive on time.
}

We consider probabilistic failures. For every edge $e \in E$, we assume there is a probability $p^e_{crash}$ that $e$ crashes as well as a probability $p^e_{omit}$ that a forwarded message on $e$ is omitted. Allowing different probabilities for the edges is useful for modeling settings in which the links are of different quality. Note that we allow ``ideal'' links with probability $0$ of failing.  Faults occur independently though some dependencies arise from our definitions and we highlight them below. In the temporary-crash model, the probability that $e$ is active at a time $i$ is $1-p^e_{crash}$. In the permanent-crash model, crashes are dependent. Consider a set of active edges $T \subseteq E$. The probability that the active edges in the next time step are $T' \subseteq T$ is $\prod_{e \in T'}(1-p^e_{crash}) \cdot \prod_{e \in (T \setminus T')} p^e_{crash}$.
We define omissions similarly. Consider a configuration $C$, active edges $T$, and let $O$ be the output of the algorithm. The probability that an omission occurs to a pair in $\zug{m,e} \in O$ is $p^e_{omit}$. Here too there is dependency between omissions and crashes: an omission can only occur on an edge that a message is sent on, thus the edge must be active. Such fault probabilities give rise to a probability distribution on $\O$, which we refer to as $\D(\O)$.  

\begin{definition}
Consider $1 \leq \ell \leq |\M|$. Let $G$ be the set of outcomes in which at least $\ell$ messages arrive on time. We define $\score(\F) = \Pr_{\pi \sim \D(\O)}[\pi \in G]$.
\end{definition}

\stam{
\begin{remark}
It is possible to generalize this definition of score by considering a {\em specification} $\S$ on outcomes. A specification can be a subset of $\O$ that indicates which outcomes are good outcomes. Alternatively, $\S$ can be thought of assigning values in $\set{0,1}$ to outcomes. A further generalization assumes $\S$ assigns richer values. In both cases we can view $\S$ as a random variable. The score can be defined using a threshold value $c \in \Real$, then it is $\Pr[\S \geq c]$, or it can be the expectation of $\S$. The choice of which definition to use depends on the practical setting. Our results in the following sections carry over to these more general definitions.
\end{remark}
}

\section{From Scoring to Markov Chain}
\label{sec:score2markov}
In this section we show how to reduce the problem of finding the score of a forwarding scheme to a reachability problem on a Markov chain. 
We start with temporary crashes and omissions. A {\em deterministic automaton} (DFA, for short) is a tuple $\D = \zug{\Sigma, Q, \delta, q_0, F}$, where $\Sigma$ is an alphabet, $Q$ is a set of states, $\delta:Q \times \Sigma \rightarrow Q$ is a transition function, $q_0 \in Q$ is an initial state, and $F \subseteq Q$ is a set of accepting states. We use $|\D|$ to denote the number of states in $\D$. An {\em automaton frame} is a DFA with no accepting states. A Markov chain is a tuple $\zug{Q, \P, q_0}$, where $Q$ is a set of states, $\P: Q \times Q \rightarrow [0,1]$ is a probability function such that for every state $q \in Q$, we have $\sum_{e = \zug{q,p} \in Q \times Q} \P[e] = 1$, and $q_0 \in Q$ is an initial state. A Markov chain induces a probability distribution on finite paths. The probability of a path $\pi = \pi_1, \ldots, \pi_n$, where $\pi_1 = q_0$ is the product of probabilities of the transitions it traverses, thus $\Pr[\pi] = \prod_{1 \leq i <n} \Pr[\zug{\pi_i, \pi_{i+1}}]$. For a bound $t \in \Nat$, we use $\Pr_{\set{\pi:|\pi|\leq t}}$ to highlight the fact that we are restricting to the probability space on runs of length at most $t$.

Consider a network $\N = \zug{V,E}$, a set of messages $\M$, a forwarding scheme $\F$, and a message $m \in \M$. We describe an automaton frame $\D_m[\N,\M, \F]$ that simulates the routing of $m$ in $\N$ using $\F$. We have $\D_m[\N,\M,\F] = \zug{(2^\M \times 2^E) \cup (E \cup \set{\bot}), V \cup E, \delta_m, s(m)}$, where we describe $\delta_m$ below. We omit $\N$, $\M$, and $\F$ when they are clear from the context. Intuitively, the subset of states $V$ model positions in the network and the subset of states $E$ are intermediate states that allow us to model omissions. When $\D_m$ is at state $v \in V$, it models the fact that $m$ is in the switch $v$. Accordingly, the initial state is $s(m)$ and the transition function $\delta_m$ simulates the forwarding scheme $\F$: every outgoing transition $\tau$ from a state $v \in V$ corresponds to forwarding rule $\varphi \rightarrow \Forw(m, e_i)$ for $m$. The transition $\tau$ is labeled by an alphabet letter $(M, T)$, where $M \subseteq \M$ models the messages in $v$'s queue, and $T \subseteq E$ models the active edges. Furthermore, we have $(M,T) \models \varphi$, thus $m$ is forwarded on the $i$-th edge leaving $v$.  We refer to the state at the end-point of the transition $\tau$ as $e \in E$, thus $e$ is the $i$-th edge leaving $v$. Recall that $e$ is used to model omission. Accordingly, it has two outgoing transitions: one directs back to $v$, and the second models a successful transmission and directs to the state that corresponds to the vertex $t(e)$. We define the transition function $\delta_m$ formally. For $e =\zug{v,u} \in E$, we have $\delta_m(e, e) = u$ and $\delta_m(e, \bot) = v$, and for $v \in V$, $M \subseteq \M$, and $T \subseteq E$, we have 
\[
\delta_m(v, (M,T)) = \begin{cases} e & \text{ if } \exists \varphi \rightarrow \Forw(m, e_i) \in \A, \ e= e_i, \text{ and } (M,T) \models \varphi \\
\zug{v,v} & \text{ otherwise}. \end{cases}
\]

Next, given a network $\N$, a set of messages $\M$, and a forwarding scheme $\F$, we construct an automaton-frame DFA $\D[\N,\M, \F]$ that simulates the runs of all the $\D_m$ frames. Consider a guarantee constant $1 \leq \ell \leq |\M|$. The constant $\ell$ determines the accepting states of $\D[\N,\M,\F]$: states in which at least $\ell$ messages arrive on time are accepting. Formally, we have $\D^\ell[\N,\M,\F] = \zug{2^E, V^{|\M|} \cup E^{|\M|}, \delta, q_0^\D, F_\ell}$, where we describe the definition of $q_0^\D$, $\delta$, and $F_\ell$ below. We omit $\N$, $\M$, $\F$, and $\ell$ when they are clear from the context. Recall that $\D$ simulates the execution of the network when routing according to $\F$. A state $\zug{v_1, v_2,\ldots, v_{|\M|}}$ in $\D$ represents the fact that, for $1 \leq i \leq |\M|$, message $m_i$ is in the switch $v_i$ and its frame is in the corresponding state, and similarly for a state in $E^{|\M|}$. Accordingly, the initial state $q_0^\D$ is $\zug{s(m_1),\ldots, s(m_{|\M|})}$ and a state is accepting iff at least $\ell$ messages arrive at their destination, thus $F_\ell = \set{\zug{v_1, \ldots, v_{|\M|}}: |\set{j : v_j = t(m_j)}| \geq \ell}$. Recall that the alphabet of a frame $\D_m$ consists of two types of letters; a letter $M \subseteq \M$ models the messages in a switch's queue and a letter $T \subseteq E$ models failures. Since in $\D$,  the messages in the queues can be induced by the positions of the frames, the alphabet of the frame $\D$ consists only of the second type of letters. Consider a state $\zug{v_1, v_2,\ldots, v_{|\M|}}$ in $\D$ and an input letter $T \subseteq E$. For $1 \leq i \leq |\M|$, let $M \subseteq \M$ be the messages at vertex $v_i$, thus $M = \set{m_j: v_j=v_i}$. Then, the $i$-th component in the next state of $\D$ is $\delta_{m_i}(v_i, (M, T))$. The definition for states in $E^{|\M|}$ is similar, though here, when an outgoing transition is labeled by a letter $O \subseteq E$, it models the messages that where successfully delivered. 

Recall that the letters in $\D[\N, \M, \F]$ model failures. We assume probabilistic failures, thus in order to reason about $\N$ we construct a Markov chain $\C[\N, \M,\F]$ on the structure of $\D[\N, \M, \F]$ by assuming a distribution on input letters. Formally, we have $\C[\N, \M, \F] = \zug{V^{|\M|} \cup E^{|\M|}, \P, q_0^\D}$, where $\tau= \zug{\overline{v}, \overline{e}} \in V^{|\M|} \times E^{|\M|}$ has a positive probability iff there exists $T \subseteq E$ such that $\delta(\overline{v}, T) = \overline{e}$, then $\P[\tau] = \prod_{e \in T} p_e \cdot \prod_{e \notin T} (1-p_e)$, and the definition of edges from states in $E^{|\M|}$ to $V^{|\M|}$ is similar. We can now specify the score of a forwarding scheme as the probability of reaching $F_\ell$ in $\C[\N, \M, \F]$.

\begin{theorem}
\label{thm:markov}
Let $\N$ be a network, $\M$ a set of messages, $\F$ be a forwarding scheme, and $1 \leq \ell \leq |\M|$ a guarantee. For a timeout $t \in \Nat$, we have that $\Pr_{\set{\pi:|\pi|\leq t}} [\set{ \pi: \pi \text{ reaches } F_\ell}]$ in $\C[\N, \M, \F]$ equals $\score(\F)$.
\end{theorem}

The construction above considers temporary crashes. Recall that in permanent crashes, once an edge crashes it does not recover. In order to reason about such crashes, we take a product of $\D$ with $2^{|E|}$. A state that is associated with a set $T \subseteq E$ represents the fact that the edges in $E \setminus T$ have crashed. Thus, input letters from such a state include only edges in $T$. 

\stam{
\begin{remark}
\label{rem:crashes}
Generalizing this idea is helpful. We can assume a {\em specification} on faults given as a DFA for each edge. For example, we can require that if an edge crashes, it stays down for at least, at most, or exactly $c$ time units. Then, we take the product of $\D$ with these automata similar to the product construction we used above.
\end{remark}
}

\section{Computing the Score of a Forwarding Scheme}
\label{sec:exact}
While Theorem~\ref{thm:markov} suggests a method to compute the score of a forwarding scheme by solving a reachability problem on the Markov chain $\C$, the size $\C$ is too big for practical purposes. In this section we reason about $\C$ without constructing it implicitly by reducing the scoring problem to \#SAT, the problem of counting the number of satisfying assignments of a Boolean formula. We proceed in two steps.

\subsubsection{Simulating executions of $\D$}
Recall that the Markov chain $\C$ shares the same structure as an automaton $\D$ whose input alphabet represents faults. We reason about $\D$ by constructing a Boolean formula $\psi$ whose satisfying assignments correspond to accepting runs of length $t$ of $\D$, which correspond in turn to ``good outcomes'' of the network, i.e., outcomes in which at least $\ell$ messages arrive on time. The crux of the construction is that the size of $\psi$ is proportional to the sum of sizes of the $\D_m$ automata that compose $\D$ rather than the product of their sizes, which is the size of $\D$. In order to ensure that the run a satisfying assignment simulates, is accepting, we need to verify that at least $\ell$ messages arrive on time. We show how to simulate a counter using a Boolean formula in the following lemma. 
\begin{lemma}
\label{lem:counter}
Consider a set $X$ of $|\M|$ variables, a truth assignment $f: X \rightarrow \set{\tt, \ff}$, and a constant $1 \leq \ell \leq |\M|$. There is a Boolean formula $CNT_\ell$ over variables $X \cup Y$ such that there is a satisfying assignment to $CNT_\ell$ that agrees with $f$ on $X$ iff $|\set{x \in X: f(x) = \tt}| \geq \ell$. The size of $Y$ is $|\M| \cdot \log \lceil \ell +1 \rceil$ and $CNT_\ell$ has linear many constraints in $|X \cup Y|$.
\end{lemma}
\begin{proof}
Let $X = \set{x_1, \ldots, x_n}$ be a set of variables that is ordered arbitrarily and $1 \leq \ell \leq n$. We simulate a Boolean circuit that has $n$ bits of input (corresponding to an assignment to the variables in $X$), counts the number of variables that are assigned $1$, and returns $1$ iff there are at least $\ell$ such variables. Since we need to count to $\ell$, we need $\lceil \log \ell \rceil$ bits, and we need $n$ copies of the bits. For $1 \leq i \leq n$, let $\overline{y_i} = \set{y^j_i: 1 \leq j  \leq \lceil \log \ell \rceil}$. We define $Y = \bigcup_{1 \leq i \leq n} \overline{y_i}$. The first copy of the counting bits is initialized to $0$, thus we have a constraint $\neg y^1_1 \wedge \ldots \wedge \neg y^{\lceil \log \ell \rceil}_1$. For $1 \leq i \leq n$, we add constraints so that if the assignment to $x_i$ is $0$, then $\overline{y_i} = \overline{y_{i+1}}$ (i.e., the counter is not incremented), and if $x_i$ is $1$, then (roughly) $\overline{y_{i+1}} = \overline{y_i} +1$. Both can be achieved with polynomial many constraints in $\lceil \log \ell \rceil$. Finally, we add constraints that require that at least one of the counters equals $\ell$.  
\end{proof}

We proceed to construct the formula $\psi$.

\begin{theorem}
\label{thm:psi}
Given a forwarding scheme $\F$ for a network $\N$, a set of messages $\M$, and two constants $t, \ell \in \Nat$, there is a Boolean formula $\psi$ such that there is a one-to-one correspondence between satisfying assignment to $\psi$ and accepting runs of $\D^\ell[\N,\M,\F]$. The size of $\psi$ is $poly(|\N|,|\F|, |\M|, t, \log \ell)$.
\end{theorem}
\begin{proof}
We use $|\M| \cdot |\N| \cdot t$ variables to simulate the execution of the underlying $|\M|$ frames. A variable of the form $x_{m,v,i}$ represents the fact that message $m$ is on switch $v$ at time $i$. We model the faults using variables: a variable $x_{e,i}$ represents the fact that $e$ is active at time $i$ and a variable $x_{e,m,i}$ represents the fact that sending message $m$ on link $e$ at time $i$ was successful. Recall that the transition function of the frames corresponds to the forwarding algorithm, which is given by a set of propositional rules. We simulate these rules using a Boolean formula over the variables. Finally, we add constraints that require that the run starts from the initial state, i.e., $x_{m, s(m), 1} = \tt$, and ends in an accepting state, i.e., $|\set{m \in \M : x_{m, t(m), t} = \tt}| \geq \ell$. For the later we use the assertion $CNT_\ell$ that is described in Lemma~\ref{lem:counter} with $X = \set{x_{m, t(m), t}: m \in \M}$.
 
We construct $\psi$ formally. Consider a network $\N = \zug{V, E}$, a set of messages $\M$, a forwarding scheme $\F$, and constants $\ell, t \in \Nat$. We describe the variables in $\psi$. For every message $m \in \M$, we have $(t+1) \cdot |\D_m|$ variables of the form $x_{m, v, i}$ and $x_{m, e, i}$, which represent respectively, the fact that $m$ is on vertex $v$ and that $m$ is send on edge $e$ at time $i$. Also, we have $2t\cdot |E|$ variables of the form $x_{e,i}$ that represent the fact that $e$ crashes at time $i$ (for odd $i$) and that an omission occurs on $e$ at time $i$ (for even $i$). We sketch the constraints in $\psi$. The first constraint requires that all messages start from their origins, thus we have $x_{m, s(m), 0} = \tt$. We simulate the transition function of $\D$ using constraints. We have $t \cdot |V| \cdot |\M|$ copies of every rule $\varphi \rightarrow \Forw(e_j, m)$. For every message $m \in \M$, vertex $v \in V$, and time $1 \leq i \leq t$, we re-write $\varphi$ as a constraint over the variables in $\psi$ by replacing appearances of $m'$ with $x_{m', v, i}$ and of $e$ with $x_{e, i}$. Then, we add a constraint to $\psi$ that requires that if $\varphi$ holds and $x_{m,v,i}$ holds, then $x_{m, e_j, i}$, thus $m$ is forwarded on $e_j$ at time $i$. The constraints that corresponds to outgoing transitions from states in $E^{|\M|}$ are similar. Finally, we require that at least $\ell$ messages arrive on time. Requiring that all messages arrive on time is easy; all we need to do is add a constraint $x_{t(m), m, t} = \tt$, for all $m \in \M$. In order to relax this constraint, we need an SMT constraint of the form $\sum_{m \in \M} x_{t(m), m, t} \geq \ell$. By Lemma~\ref{lem:counter}, this constraint can be specified as a Boolean constraint. 

Consider a satisfying assignment $f$ to $\psi$ and let $r$ be the corresponding rejecting run of $\D$. The probability of $f$ is the product of probabilities of letters it uses and it clearly coincides with the probability of $r$ in $\C$. Since there is a one to one correspondence between satisfying assignments and rejecting runs, we have $\score(\F) = \sum_{f \in SAT(\psi)} \Pr[f]$, and we are done.\hfill\qed
\end{proof}

\subsubsection{Reasoning about $\C$ using $\psi$}
\label{sec:WMC}
Recall that in Theorem~\ref{thm:markov}, we reduce the problem of scoring a forwarding scheme to the problem of finding the probability of reaching the accepting states in $\C$ in $t$ iterations. By Theorem~\ref{thm:psi} above, a satisfying assignment $f$ to $\psi$ corresponds to such an execution $r$. We think of $f$ as having a probability, which is $\Pr[r]$. Let $SAT(\psi)$ be the set of satisfying assignments to $\psi$. We have established the following connection: $\score(\F) = \sum_{f \in SAT(\psi)} \Pr[f]$. 

Recall that \#SAT is the problem of counting the number of satisfying assignments of a Boolean formula. The counting problem in the right-hand side of the equation above is a {\em weighted-model counting} (WMC, for short) problem, which generalizes \#SAT. The input to WMC is a Boolean formula $\varphi$ and a weight function $w$ that assigns to each satisfying assignment a weight, and the goal is to calculate $\score(\varphi) = \sum_{f \in SAT(\varphi)} w(f)$.  \#SAT is a special case in which the weight function is $w \equiv 1$, thus all assignments get weight $1$. In order to distinguish between the two problems, we sometimes refer to \#SAT as {\em unweighted model counting} (UMC, for short). 

The last step in our solution adjusts $\psi$ to fit in the framework of \cite{CFMV15} and use the reduction there from WMC to UMC. Their framework deals with weight functions of a special form: each literal has a probability of getting value true and the literals are independent. So the weight of an assignment is the product of the literals' probabilities. Accordingly, they call this fragment {\em literal-weighted} WMC. Formally, we have a probability function $\Pr[l]$, for every literal $l$ in $\psi$. We define $w(f) = \prod_{l: \sigma(l) = \tt} \Pr[l] \cdot \prod_{l: \sigma(l) = \ff} (1-\Pr[l]))$, and $\score(\psi) = \sum_{f \in SAT(\psi)} w(f)$. 
\begin{theorem}
\label{thm:literal WMC}
Consider the WMC  instance $\zug{\psi,w}$, where $\psi$ is the Boolean formula obtained in Theorem~\ref{thm:psi} and, for $f \in SAT(\psi)$ with corresponding execution $r$, we have $w(f)=\Pr[r]$. There is a literal-weighted WMC $\zug{\psi', w'}$ and a factor $\gamma$ such that $\gamma \cdot \score(\psi') = \score(\psi)$ and $\psi'$ is polynomial in the size of $\psi$.
\end{theorem}
\begin{proof}
We start with temporary crashes and omits.
Recall that there are two types of variables in $\psi$; variables of the form $x_{m,v,i}$ that simulate the runs of the underlying automata and variables of the form $x_{e,i}$ that represent the fact that a fault occurs in $e$ (crashes for odd $i$ and omissions for even $i$). Since the automata are deterministic, the values of the first type of variables is determined by the second type of variables. A first attempt to define the weights of the $x_{e,i}$ variables would be to set them to $p^e_{crash}$ and $p^e_{omit}$, respectively. However, this definition fails as there is dependency between crashes and omits; an omit cannot occur on an edge that crashes. In the following, we introduce new variables to correct the dependencies. 

It is convenient to add a variable $fr_{e,i}$ that gets value true when one of the messages is forwarded on $e$ at time $i$, thus an omission can occur only if $fr_{e,i}=\tt$. Note that it is implicit that $fr_{e,i}=\tt$ only when $e$ does not crash. Let $i$ be even, and recall that $x_{e,i}=\tt$ when $e$ exhibits an omission. The behavior we are expecting is $\Pr[x_{e, i} = \tt | fr_{e, i} = \tt] = p^e_{omit}$ and $\Pr[x_{e, i} = \tt | fr_{e, i} = \ff] = 0$. In order to model this behavior, we multiply the score of $\psi'$ by $\gamma$, add two independent variables $a_{e,i}$ and $b_{e,i}$ with respective weights $a$ and $b$, which we calculate below, and constraints $a_{e, i} = x_{e, i} \land fr_{e, i}$ and $b_{e, i} = \neg x_{e, i} \land \neg fr_{e, i}$. Recall that $\Pr[x_{e, i} = \tt | fr_{e, i} = \tt]$ should equal $p^e_{omit}$. In that case, we have $a_{e,i} = \tt$ and $b_{e,i} = \ff$ with probability $a \cdot (1-b)$. Thus, we have $p^e_{omit} = \gamma \cdot a\cdot (1-b)$. We do a similar calculation for the three other cases to obtain two other equations: $1-p^e_{omit} = \gamma \cdot (1-a) \cdot (1-b)$ and $1 = \gamma \cdot (1-a) \cdot b$. Thus, we define $a = p^e_{omit}$, $b=\frac{1}{2-p^e_{omit}}$, and $\gamma^{-1} = (1-a) \cdot b$.

In the permanent-crash model, there are dependencies between crashes; once an edge crashes it cannot recover. We use a similar technique to overcome these dependencies. We introduce two new variables $c_{e, i}$ and $d_{e,i}$ with weights $c=p^e_{crash}$ and $d=\frac{1}{2-p^e_{crash}}$, and constraints $c_{e,i} = x_{e, i} \land x_{e, i-2}$ and $d_{e, i} = x_{e, i} \land \neg x_{e, i-2}$. Additionally, $x_{e, 1}$ is assigned weight $p^e_{crash}$ as it has no dependencies. Note that $\psi'$ is of size polynomial in $\psi$ as we have added at most $4tE$ new variables and constraints, the largest of them ($fr_{e, i} = \bigvee_{m \in \M} x_{m, e, i}$) having size $|M|$. As shown above, each pair $\zug{a_{e,i}, b_{e,i}}$ contributes $\gamma^{-1}$ to the normalization factor. Similarly, each pair $\zug{c_{e,i}, d_{e,i}}$ contributes $((1-c)\cdot d)^{-1}$.\hfill\qed
\end{proof}

Finally, we use the reduction from literal-weight WMC to UMC as described in \cite{CFMV15}, thus we obtain the following. 

\begin{theorem}
\label{thm:sharpSAT}
The problem of scoring a forwarding scheme is polynomial-time reducible to \#SAT.
\end{theorem}

\section{Computational Complexity}
We study the computational complexity of finding the score of a forwarding scheme. We show that it is \#P-complete by showing that it is equivalent to the problem of counting the number of satisfying assignments of a Boolean formula (a.k.a the \#SAT problem). 
\begin{theorem}
\label{thm:sharp-P}
The problem of computing the score of a forwarding scheme is \#P-Complete.
\end{theorem}
\begin{proof}
The upper bound follows from Theorem~\ref{thm:sharpSAT}. For the lower bound, we reduce \#3SAT, the problem of counting the number of satisfying assignments of a 3CNF formula, to the problem of finding the score of a forwarding scheme. Consider an input 3CNF formula $\psi = C_1 \wedge \ldots \wedge C_k$ over a set $X$ of $n$ variables. We construct a network $\N$ with $n+k$ messages, a forwarding scheme $\F$, and $t, \ell \in \Nat$, such that the number of satisfying assignments to $\psi$ is $(1-\score(\F)) \cdot 2^n$. 

We have two types of messages; {\em variable messages} of the form $m_x$, for $x \in X$, and {\em clause messages} of the form $m_C$, where $C$ is a clause in $\psi$. A variable message $m_x$ has two possible paths it can traverse $\pi_x$ and $\pi_{\neg x}$, where the probability of traversing each path is $0.5$. We achieve this by using the hot-potato algorithm of Example~\ref{ex:hot-potato}, using $\pi_x$ as the first-choice path for $m_x$ and $\pi_{\neg x}$ as the second-choice path, and having the first edge on $\pi_x$ crash with probability $0.5$ and all other edges  cannot crash.  There is a clear one-to-one correspondence between outcomes and assignments to the variables: an outcome $\tau$ corresponds to an assignment $f: X \rightarrow \set{\tt, \ff}$, where $f(x) = \tt$ if $m_x$ traverses $\pi_x$ in $\tau$ and $f(x) = \ff$ if $m_x$ traverses $\pi_{\neg x}$ in $\tau$. Since crashes in times later than $0$ do not affect the choice of $m_x$, we have $\Pr[\text{outcomes with } \pi_x] = \Pr[\text{outcomes with } \pi_{\neg x}]= 0.5$, thus the probability of every assignment is $1/2^n$. 


Finally, we associate satisfying assignments with bad outcomes. A bad outcome is an outcome in which no message arrives on time, thus $\ell = 1$. Both paths for the variable messages are longer than the timeout $t$, so these messages miss the timeout in any case. Each clause message $m_C$ has a unique path $\pi_C$ and its length is $t$. Let $l \in \set{x, \neg x}$ be a literal in $C$. Then, $\pi_C$ intersects the path $\pi_l$ in exactly one edge $e$. The paths are ``synchronized'' such that if $m_x$ chooses $\pi_l$, then both $m_x$ and $m_C$ reach the origin of $e$ at the same time. Since $m_x$ has precedence over $m_C$, it will traverse $e$ first, making $m_C$ wait at $s(e)$ for one time unit and causing it to miss the timeout (recall that $|\pi_C| = t$). Note that $m_C$ misses the timeout iff one of the literals in it gets value $\tt$. Thus, an outcome in which all clause messages miss the timeout, i.e., a bad outcome, corresponds to a satisfying assignment to $\psi$, and we are done.

We describe the formal details of the network. Recall that our goal is to synchronize between the variable messages and the clause messages. The initial vertex for the path of a variable $x_i$ has two outgoing edges $e^i_{pos}$ and $e^i_{neg}$. The preference of $e^i_{pos}$ is higher than $e^i_{neg}$, and the probability that $e^i_{pos}$ crashes is $1/2$. The probability of crashes and omissions for all other edges in the network is $0$. If $e^i_{pos}$ crashes, $m_{x_i}$ travels on the path $\pi_{x_i}$ and otherwise it travels on the path $\pi_{\neg x_i}$, and we describe the two paths below. 

Consider a clause $C_j$. We describe the path $\pi_j$ of length $t = 4k +1$ on which the corresponding clause message travels. The path is partitioned into $k+1$ segments, where each segment has $4$ edges apart from the $0$-th segment that has one edge. We assume some arbitrary order on variables. Assuming the variables that appear in $C_j$ are $x_{i_1}, x_{i_2}$, and $x_{i_3}$, where $x_{i_1} < x_{i_2}<x_{i_3}$, then exactly one of the paths of the corresponding variable messages cross the $j$-th segment in $C_j$'s path.  For $x \in \set{x_{i_1}, x_{i_2},x_{i_3}}$, if $x$ appears in positive form, this is the  path $\pi_x$ and otherwise it is the  path $\pi_{\neg x}$. For $l = 1,2,3$, the message that corresponds to $x_{i_l}$ shares the $(l+1)$-th edge with $m_{C_j}$. The last edge takes care of cases in which $x_{i_3}$ is the first variable in the clause $C_{j+1}$, and allows the message time to ``skip'' to the other path. The paths of the message variables have intermediate vertices and edges so that this synchronization is guaranteed as well as other vertices and edges that guarantee that the length of the paths exceed $t$. 
\hfill\qed
\end{proof}

\section{Estimating the Score of a Forwarding Scheme}
\label{sec:estimate}
In this section we relax the requirement of finding an {\em exact} score and study the problem of estimating the score. We study probabilistic algorithms that with high probability return a score that is close to the exact score.

\subsubsection{Iterative counting approach}
We build on the counting method developed in Section~\ref{sec:exact}. A first attempt to estimate the score would be to feed the Boolean formula $\psi'$ we develop there into a tool that approximately solves \#SAT. However, this attempt fails as the reduction of \cite{CFMV15} from weighted to unweighted counting produces an instance that is particularly hard to solve for such solvers. In order to use the literature on approximate counting, we must develop a different technique. We take advantage of the fact that in practice, the probability of failures is very small. Thus, the executions that include many faults have negligible probability. We find an approximate score of a forwarding scheme in an iterative manner. We start with a score of $0$ and uncertainty gap $1$, and iteratively improve both. We allow only permanent edge crashes in this approach and we require all edges to have the same probability. In each iteration we allow exactly $k$ crashes. Calculating the probability of all outcomes with $k$ crashes is not hard. 
\begin{lemma}
\label{lem:all-prob}
The probability of all outputs with exactly $k$ crashes is ${|E| \choose k} \cdot (1-p_{crash})^{(|E|-k) \cdot t} \cdot (1 - (1-p_{crash})^t)^k$.
\end{lemma}
\begin{proof}
We first choose the $k$ edges that crash. The probability that the other edges do not crash is $(1-p)^{(|E|-k) \cdot t}$. The probability that an edge does not crash is $(1-p)^t$. Thus, the probability that it crashes at some time is $(1-(1-p)^t)$, and we take the product for the $k$ edges that do crash.\hfill\qed
\end{proof}

We find the probability of the ``good outcomes'' with $k$ crashes using a counting method, add to the score of the scheme and update the uncertainty gap by deducting the probability of the bad outcomes. 
We use the weighted counting framework of \cite{CFMSV14} (which is not weighted-literal WMC). Restricting to $k$ crashes has two advantages, which significantly speed up the counting. First, the solution space is significantly reduced. More importantly, we use the fact that the probabilities of the outcomes do not vary too much. The running time of the method of \cite{CFMSV14} depends on a given estimation of the ratio between the weight of the maximal weighted satisfying assignment and the minimal weighted one, which the authors refer to as the {\em tilt}.  
\begin{lemma}
\label{lem:tilt}
$tilt \leq (1-p_{crash})^{k \cdot t}$.
\end{lemma}
\begin{proof}
Note that the probability of an outcome with a crash at time $i$ is greater than the probability of the same outcome only with the crash occurring at time $i+1$. Indeed, in the second outcome, the edge has to ``survive'' the $i$-th time slot, thus the probability of the outcomes differ by a factor of $(1-p_{crash})$. Thus, having all crashes occur at time $0$ is an upper bound on the bad outcome with highest probability. Similarly, having all crashes occur at time $t$ is a lower bound on the minimal-weighted outcome. Since all other edges do not crash, we have that $tilt \leq (1-p_{crash})^{k \cdot (t-1)}$.\hfill\qed
\end{proof}

We describe the pseudo code of the approach below.

\begin{center}
\begin{algorithmic}
\algnotext{EndIf}
\algnotext{EndWhile}
\Require A network $\N=\zug{V, E}$, a set of messages $\M$, a forwarding scheme $\F$, constants $t, \ell \in \Nat$, the probability of a permanent crash $p_{crash}$, and $\epsilon > 0$.

\Ensure An additive $\epsilon$-approximation of $\score(\F)$.

\State $uncertainty = 1, score=0, k=0$.
\While{$uncertainty > \epsilon$}
\State $all \gets$ Probability of all outcomes with $k$ crashes.
\State $bad \gets \Call{CalcBadProb}{\N, \M, \S, t, \ell, k}$
\State $uncertainty \ -\!= \ all; \ score \ +\!= \  (all-bad);\ k\ +\!+;$
\EndWhile
\State \Return $score$
\end{algorithmic}
\end{center}

\subsection{A Monte-Carlo Approach}
The Monte-Carlo approach is a very simple and well-known approach to reason about reachability in Markov chains. It performs well in practice as we elaborate in Section~\ref{sec:eval}. We perform $n$ probabilistic simulations of the execution of the Markov chain $\C$ for $2t$ iterations, where $t$ is the timeout and $n$ is a large number which we choose later. In each simulation, we start from the initial state of $\C$. At each iteration we probabilistically choose an outgoing edge and follow it. If we reach a state in $F_\ell$, we list the experiment as $1$, and otherwise as $0$. We use $y_1, \ldots, y_n$ to refer to the outcomes of the experiments, thus $y_i \in \set{0,1}$. Let $r$ be the number of successful experiments. We return $r/n$. We use Hoeffding's inequality to bound the error: $\Pr[\frac{1}{n}\sum_{i=1}^n y_i - \score(\F) \geq \epsilon] \leq e^{-2n\epsilon^2}$. Thus, we choose $n$ so that given requirements on the error and confidence are met.

\section{Evaluation}
\label{sec:eval}
In this section we evaluate the techniques to compute the exact and approximate score of a forwarding scheme. We compare the scalability of these approaches. Our counting techniques rely on black-boxes that count the number of satisfying assignments of a SAT formula. We used sharpSAT \cite{Thu06}  to exactly solve \#SAT and WeightMC \cite{CFMSV14} to approximately solve weighted \#SAT.  Our implementations are in Python and we ran our experiments on a personal computer; an Intel Core i3 quad core 3.40 GHz processor.

\subsubsection{Generating a setting}
We evaluate the algorithm on networks that were generated randomly using the library Networkx \cite{HSS08}. 
We fix the number of vertices, edges, and messages and generate a random directed graph. We consider relatively dense graphs, where the number of edges are approximately $2.5$ times the number of vertices. Once we have a graph, we randomly select a source and a target for each message. Recall that a forwarding scheme has three components: the forwarding algorithm, message priorities, and edge priorities for each message. 

The forwarding algorithm we use is the ``Hot-potato'' algorithm, which is described in Example~\ref{ex:hot-potato} and has some error-handling capabilities. We choose the message priorities arbitrarily, and we choose the edge preference as follows. We follow a common practice in generating TT-schedules in which we restrict messages to be scheduled on few predefined paths from source to target \cite{steiner2010evaluation,pozo2015decomposition}. For each message, we select a ``first-choice'' path $\pi_m$ using some simple heuristic like taking the shortest path between $s(m)$ and $t(m)$, and a ``fall-back'' path from each vertex on $\pi_m$ to $t(m)$. The collection of fall-back paths form a DAG with one sink $t(m)$. 
This restriction significantly shrinks the formula $\psi$ that we construct.
\stam{
This restriction is helpful in our setting as the automaton $\D_m$ for a message $m \in \M$ can be restricted to the positions in which $m$ can visit, and as a result the variables and constraints in the formula $\psi$ can also be reduced. 
For each vertex $v$ on $\pi_m$ let $e$ and $e'$ be the the next edges on the first- and second-paths, respectively. Then, we have $e' \prec_m^v e$, and the other edges in $out(v)$ have no preference. 
}
 We assume permanent crashes, and set the probabilities of a crash and an omission uniformly in the network to be $0.01$.
 This is a very high probability for practical uses, but we use it because it is convenient to evaluate the calculation methods with a high probability, and the actual score of the forwarding scheme is less important to us.
All results have been averaged over 3-5 runs. Each program times out after 1 hour, returning ``timeout'' if it has not terminated by then.

\paragraph{Execution time measurements}
We have implemented the exact and estimating approaches that are described in Sections~\ref{sec:exact} and~\ref{sec:estimate}. The running times are depicted in Figure~\ref{fig:runtime}. We note that it is unfair to compare the exact method to the estimation ones, and we do it nonetheless as it gives context to the results. The sharpSAT tool performs well (even better than the approximation tools) for small instances. But, the jump in running time is sudden and occurs for networks with $7$ nodes, where the running time exceeded an hour.

\stam{
 \begin{figure}
  \begin{minipage}[b]{0.5\linewidth}
\includegraphics[height=5cm]{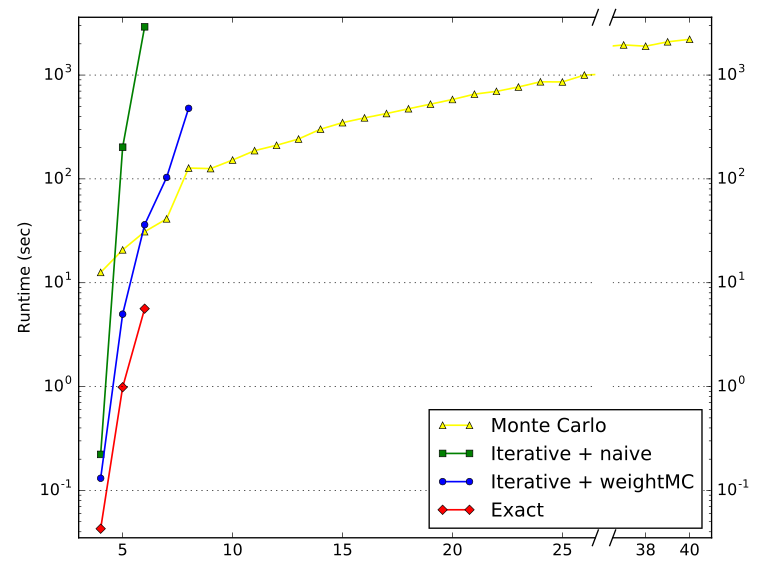}
\caption{\small The running time of the approaches on increasing-sized networks.}
\label{fig:runtime}
\end{minipage}
  \hspace{0.6cm}
  \begin{minipage}[b]{0.4\linewidth}
        \centering
 \begin{tabular}{|c||c|c|c|}
 \hline
 Num. of nodes&  Exact &  Monte-Carlo & Error \\
 \hline\hline
 $4$  & $0.998$ & $0.998$ & $0.0002$\\
 \hline
 $5$  & $0.965$ & $0.963$ & $0.001$\\
 \hline
 $6$  & $0.967$ & $0.968$ & $0.0005$ \\
 \hline
 \end{tabular}
\caption{\small Comparison of the exact score with the one obtained by the Monte-Carlo approach.}
\label{tab:score}
  \end{minipage}%
\end{figure}
}

\begin{figure}
\begin{floatrow}
\ffigbox{%
\includegraphics[height=4cm]{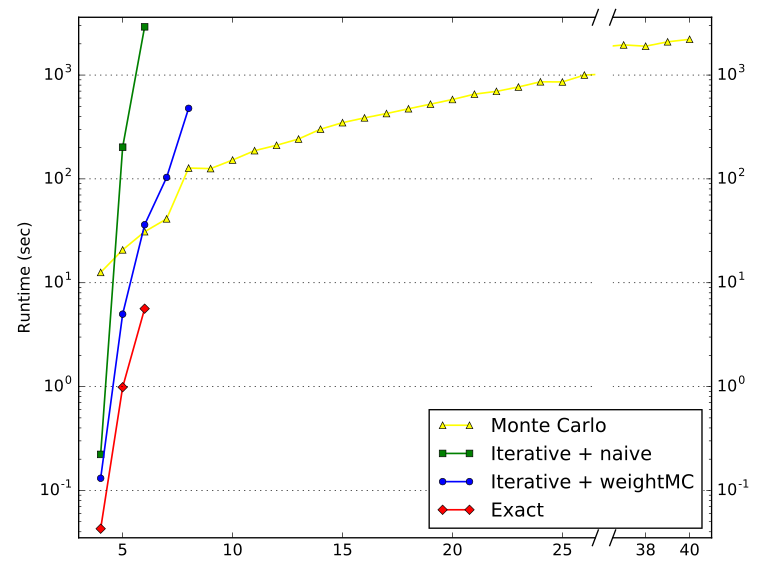}
}{%
\caption{\small The running time of the approaches on increasing-sized networks.}
\label{fig:runtime}
}
\capbtabbox{%
\hspace{0.2cm}
 \begin{tabular}{|c||c|c|c|}
 \hline
 Num. of nodes&  Exact &  Monte-Carlo & Error \\
 \hline\hline
 $4$  & $0.998$ & $0.998$ & $0.0002$\\
 \hline
 $5$  & $0.965$ & $0.963$ & $0.001$\\
 \hline
 $6$  & $0.967$ & $0.968$ & $0.0005$ \\
 \hline
 \end{tabular}
}{%
\caption{\small Comparison of the exact score with the one obtained by the Monte-Carlo approach.}
\label{tab:score}
}
\end{floatrow}
\end{figure}

For estimating the score, we have implemented two approaches; an iterative approach and a Monte-Carlo approach. Recall that the crux in the first approach is computing the probability of bad outcomes with exactly $k$ crashes. We use two techniques; the tool weightMC \cite{CFMSV14} as well as a naive counting method: we iteratively run Z3 \cite{MB08} to find an assignment and add its negation to the solver so that it is not found again. We combine the naive approach with an optimization that is similar to the one that was shown to be helpful in \cite{AGR16}, but we find it is not helpful in our setting.

Finally, we implemented a Monte-Carlo approach in Python using randomization functions from the Numpy library. 
We ran the simulations on $4$ threads, which we found was an optimal number for our working environment. We evaluated the Monte Carlo approach using an error $\epsilon=0.01$, and a confidence of $\delta=0.99$. 

The leading estimation method is the Monte-Carlo approach, which scales quite well; in reasonable time, it can calculate the score of moderate sized networks and shows a nice linear escalation with the network growth. It is somewhat frustrating that this simple approach beats the approaches that rely on counting hands down as a significant amount of work, both theoretical and in terms of optimizations, has been devoted in them. As mentioned earlier, the research on SAT counting is still new and we expect improvements in the tools, which will in turn help with our scalability.

\paragraph{Evaluating the approximation}
Apart from the theoretical interest in an exact solution, it can serve as a benchmark to evaluate the score the estimation methods output. In Table~\ref{tab:score}, we compare the scores obtained by the exact solution and by the Monte-Carlo solution and show that the error is well below our required error of $0.01$.

\stam{
The counting approaches still have their benefits; the weighted model counting approach provides an exact calculation of the score, and can be used as a sanity check for small instances as we do in Figure~\ref{??}. Also, having more than one method of counting helps in finding bugs (e.g., we observed numerical instabilities in the WMC approach when comparing with the naive counting approach). Finally, we expect these counting tools to improve over time, and since we use them as a black-box, with their improvement, our approach will also improve.
}

\section{Discussion}
\label{sec:disc}
We introduce a class of forwarding schemes that are capable of coping with faults and we reason on the predictability of a forwarding scheme. We study the problem of computing the score of a given a forwarding scheme $\F$ in a network $\N$ subject to probabilistic failures, namely the probability that at least $\ell$ messages arrive on time when using $\F$ to forward messages in $\N$. We reduce the problem of scoring a forwarding scheme to \#SAT, the problem of counting the number of satisfying assignments of a Boolean formula. Our reduction goes through a reachability problem on a succinctly represented Markov chain $\C$. The Boolean formula we construct simulates the executions of $\C$. We considered a class of forwarding schemes that operate in a network with a notion of global time and two types of faults; edge crashes and message omissions. Our solution is general and allows extensions in all three aspects. We can add features to our forwarding scheme such as allowing ``message waits'' (as was mentioned in Example~\ref{ex:TT}) or even probabilistic behavior of the switches as long as the forwarding scheme is represented by propositional rules in the switches, we can support asynchronous executions of the switches (which requires a careful definition of ``timeout''), and we can support other faults like ``clock glitches'' in which a message arrives at a later time than it is expect to arrive. Our work on reasoning about Markov chains with the ``product-like'' structure of $\C$ is relevant for other problems in which such structures arise. For example in reasoning about concurrent probabilistic programs \cite{Var85}, where $\C$ simulates the execution of concurrent programs modeled using automata. 

\subsubsection*{Acknowledgments}
We thank Kuldeep Meel for his assistance with the tools as well as helpful discussions.

\small
\bibliographystyle{abbrv}
\bibliography{../ga}
\end{document}